\documentclass[a4paper]{article}

\usepackage{url}
\usepackage{amsmath,amssymb,amsthm}
\usepackage{tikz}
\usepackage[english]{babel}
\usepackage[colorlinks=true,citecolor=black,linkcolor=black]{hyperref}
\usepackage{todonotes}
\usepackage{xspace} 
\usepackage{fancybox}			
\usepackage{mathtools}	
\usepackage[small,bf]{caption} 
\usepackage{a4wide}

\newtheorem{theorem}{Theorem}
\newtheorem{proposition}{Proposition}
\newtheorem{corollary}{Corollary}

\newtheorem{lemma}{Lemma}

\newcommand{\np}{\mathsf{NP}}
\newcommand{\wone}{\mathsf{W[1]}}
\newcommand{\wtwo}{\mathsf{W[2]}}
\newcommand{\wpe}{\mathsf{W[P]}}
\newcommand{\xp}{\mathsf{XP}}
\newcommand{\fpt}{\mathsf{FPT}}

\newcommand{\p}{\mathsf{P}}

\newcommand{\bigo}{\mathcal{O}}
\newcommand{\bigos}{\mathcal{O}^*}

\newcommand{\apx}{\mathsf{APX}}

\newcommand{\I}{\mathcal{I}}
\newcommand{\M}{\mathcal{M}}
\renewcommand{\S}{\mathcal{S}}
\newcommand{\T}{\mathcal{T}}

\def\GM{\textsc{Graph Motif}\xspace}
\def \pmogm {\textsc{Module Graph Motif}\xspace}
\def \pegm {\textsc{Exact Graph Motif}\xspace}

\def \plmogm {\textsc{List-Colored Module Graph Motif}\xspace}
\def \pmagm {\textsc{Max Graph Motif}\xspace}
\def \pmisgm {\textsc{Min Substitute Graph Motif}\xspace}

\def \pmisc {\textsc{Min Set Cover}\xspace}
\def \pxtc {\textsc{X3C}\xspace}
\def \pmais {\textsc{Max Independent Set}\xspace}

\DeclareMathOperator{\adj}{adj}
\DeclareMathOperator{\child}{Child}

\tikzstyle{bigvertex}=[circle,fill=black!0,minimum size=15pt,inner sep=0pt]
\tikzstyle{rect}=[rectangle, rounded corners]
\tikzstyle{edge} = [draw,-,rounded corners=8pt]
\tikzstyle{vertex}=[circle, draw, inner sep=2pt, minimum width=4pt]
\usetikzlibrary{trees,decorations,decorations.pathmorphing,decorations.pathreplacing}

\newcommand{\PbOpt}[3]{%
\begin{center}
  \begin{tabular}{|l|}%
  \hline
    \begin{minipage}[c]{.95\textwidth}
    \smallskip%
      \par\noindent%
      #1:
      \par\noindent%
      $\bullet$
      \textbf{\textsf{Input}}: #2%
      \par\noindent%
      $\bullet$
      \textbf{\textsf{Output}}: #3%
      \smallskip%
      \par\noindent%
    \end{minipage}
    \\\hline
  \end{tabular}%
\end{center}
}%

\newcommand{\Pb}[4]{%
\begin{center}
  \begin{tabular}{|l|}%
  \hline
    \begin{minipage}[c]{.95\textwidth}
      \smallskip%
      \par\noindent%
      #1:
      \par\noindent%
      $\bullet$
      \textbf{\textsf{Input}}: #2%
      \par\noindent%
      $\bullet$
      \textbf{\textsf{#4}}:
      #3%
      \par\noindent%
    \end{minipage}
  \\\hline
  \end{tabular}%
\end{center}
}%

\newcommand{\PbDec}[3]{
	\Pb{#1}{#2}{#3}{Question}
}

\newlength{\parindentsave}

\begin{document}

\title{Some results on more flexible versions of \GM\footnote{An extended abstract of this paper appeared in the proceedings of the 7th international Computer Science Symposium in Russia (CSR 2012), LNCS 7353, pp. 278-289.}}

\author{%
  Romeo Rizzi\footnote{Department of Computer Science, University of Verona - Verona, Italy. \texttt{Romeo.Rizzi@univr.it}} \and 
  Florian Sikora\footnote{LAMSADE - CNRS UMR 7243, PSL, Universit\'e Paris-Dauphine (France). \texttt{florian.sikora@dauphine.fr}}  
 }


\date{}


\maketitle

\begin{abstract}
The problems studied in this paper originate from \textsc{Graph Motif}, a problem introduced in 2006 in the context of biological networks. Informally speaking, it consists in deciding if a multiset of colors occurs in a connected subgraph of a vertex-colored graph. Due to the high rate of noise in the biological data, more flexible definitions of the problem have been outlined. We present in this paper two inapproximability results for two different optimization variants of \textsc{Graph Motif}: one where the size of the solution is maximized, the other when the number of substitutions of colors to obtain the motif from the solution is minimized. We also study a decision version of \textsc{Graph Motif} where the connectivity constraint is replaced by the well known notion of graph modularity. While the problem remains $\np$-complete, it allows algorithms in $\fpt$ for biologically relevant parameterizations.

\end{abstract}


\section{Introduction}

A recent field in bioinformatics focuses in biological networks, which represent interactions between different elements (\textit{e.g.} between amino acids, between molecules or between organisms)~\cite{Alm2003}. Such a network can be modeled by a vertex-colored graph, where nodes represent elements, edges represent interactions between them and colors give functional informations on the graph nodes. 
Using biological networks allows a better characterization of species, by determining small recurring subnetworks, often called \textit{motifs}. Such motifs can correspond to a set of nodes realizing a same function, which may have been evolutionary preserved~\cite{Sharan2006}. It is thus crucial to determine these motifs to identify common elements between species and transfer the biological knowledge.

Historically, motifs were defined by a set of nodes labels with the addition of a given topology (\textit{e.g.} a path, a tree, a graph). The corresponding algorithmic problem consist to find an occurrence of the motif in the network which respect both the label set and the given topology. This leads to problems roughly equivalent to subgraph isomorphism, a computationally difficult problem. However, in metabolic networks, similar topology can represent very different functions~\cite{Lacroix2006}. Moreover, in protein-protein interactions (PPI) networks, informations about the topology of motifs is often missing~\cite{Bruckner2009}. There is also a high rate of false positive and false negative in such networks~\cite{Edwards2002}. Therefore, in some situations, topology is irrelevant, which leads to search for \emph{functional} motifs instead of \emph{topological} ones. In this setting, we still ask for the conservation of the node labels, but we replace topology conservation by the weaker requirement that the subnetwork should form a connected subgraph of the target graph. This approach was proposed by Lacroix \textit{et al.}, defining \pegm~\cite{Lacroix2006}.

\PbDec{\pegm}{A graph $G = (V,E)$, a set of colors $C$, a function $col : V \rightarrow C$, a multiset $M$ over $C$, an integer $k$.}{Does there exist a subset $V' \subseteq V$ such that (i) $|V'| = k$, (ii) $G[V']$ is connected, and (iii) $col(V')=M$.}


In the following, the motif is said colorful if $M$ is a set (it is a multiset otherwise). Note that this problem also has application in the context of mass spectrometry~\cite{Bocker2009}, and may be used in social or technical networks~\cite{Betzler2008,Sikora2011}.

\paragraph{Preliminaries}

A parameterized problem $(\I,k)$ is said \textit{fixed-parameter tractable} (or in the class $\fpt$) with respect to a parameter $k$ if it can be solved with an exact algorithm with time complexity $f(k)\cdot|\I|^c$, where $f$ is any computable function and $c$ is a constant (one can see \cite{Downey1999,Niedermeier2006}). Such algorithms are useful for $\np$-hard decision problems, thus $f$ will be exponential (but an important goal is to determine $f$ as small as possible). Parameterized algorithms do not necessarily exist, it is possible to prove that a problem do not have a parameterized algorithm via $fpt$-reductions. The parameterized complexity hierarchy is composed of the classes $\fpt \subseteq \wone \subseteq \wtwo \subseteq \dots \subseteq \wpe \subseteq \xp$. The class $\xp$ contains problems solvable in $f(k)\cdot|\I|^{g(k)}$, where $f$ and $g$ are unrestricted functions. We refer to \cite{Downey1999} for precises definitions of the $\mathsf{W}$-classes. All inclusions are assumed to be proper (but only $\fpt \neq \xp$ is known). Therefore, it is unlikely to find parameterized algorithm for a problem that is hard in $\wone$ under $fpt$-reduction (\textsc{Clique} parameterized by the size of the solution is one of them for example). A \textit{kernel} is an equivalent instance of the input obtained in polynomial time with a size bounded by a function only depending of the parameter $k$~\cite{Niedermeier2006}.

Given an instance $\I$ of an optimization problem, we use $opt(\I)$ to denote the optimum value of $\I$ and
$val(\I,S)$ to denote the value of a feasible solution $S$ of
instance $\I$. The {\em performance ratio\/} of $S$ (or {\it
approximation factor}) is
$r(\I,S)=\max\left\{\frac{val(\I,S)}{opt(\I)},
\frac{opt(\I)}{val(\I,S)}\right\}$. The  {\em error} of $S$,
$\varepsilon(\I,S)$, is defined by $\varepsilon(\I,S)= r(\I,S)-1$.
For a function $f$, an algorithm  is a {\it
polynomial-time $f(n)$-approximation\/}, if for every instance $\I$ of the problem,
it returns in polynomial time a solution $S$ such that $r(\I,S) \leq f(|\I|).$

The notion of an $E$-reduction ({\it error-preserving} reduction)
was introduced in~\cite{Khanna1994} by Khanna \textit{et al.} A problem $\Pi$
is called {\it $E$-reducible} to a problem $\Pi'$, if there exist
polynomial time computable functions $f$, $g$ and a constant
$\beta$ such that:
\begin{itemize}
\item $f$ maps an instance $\I$ of $\Pi$ to an instance $\I'$ of $\Pi'$
such that $opt(\I)$ and $opt(\I')$ are related by a polynomial
factor, \textit{i.e.} there exists a polynomial $p(n)$ such that
$opt(\I')\leq p(|\I|) opt(\I)$, 
\item $g$ maps solutions $S'$ of $\I'$
to solutions $S$ of $\I$ such that $\varepsilon(\I,S)\leq \beta
\varepsilon(\I',S')$.
\end{itemize}

An important property of an $E$-reduction is that it can be applied
uniformly to all levels of approximability; that is, if $\Pi$ is
$E$-reducible to $\Pi'$ and $\Pi'$ belongs to $\cal{C}$ then $\Pi$
belongs to $\cal{C}$ as well, where $\cal{C}$ is a class of
optimization problems with any kind of approximation guarantee
(see also \cite{Khanna1994}).

\paragraph{Previous results}

Not surprisingly, \pegm remains \sloppy $\np$-complete, even under strong restrictions (when $G$ is a bipartite graph with maximum degree 4 and $M$ is built over two colors only~\cite{Fellows2007}, or when $M$ is colorful and $G$ is a rooted tree of depth 2~\cite{Ambalath2010} or a tree of maximum degree 3~\cite{Fellows2007}). However, for general trees and multiset motifs, the problem can be solved in $\bigo(n^{2c+2})$ time, where $c$ is the number of distinct colors in $M$, while being $\wone$-hard for the parameter $c$ \cite{Fellows2007}. We also point out that the problem can be solved in polynomial time if the number of colors in $M$ is bounded and if $G$ is of bounded treewidth~\cite{Fellows2007}. It is also polynomial if $G$ is a caterpillar~\cite{Ambalath2010}, or if the motif is colorful and $G$ is a tree where the colors appears at most twice. This last result is mentioned in~\cite{Dondi2011} and can be retrieved by an easy transformation to a \textsc{2-SAT} instance (chapter~4 of~\cite{Sikora2011}).

The difficulty of this problem is counterbalanced by its fixed-parameter \sloppy tractability when the parameter is $k$, the size of the solution \cite{Lacroix2006,Fellows2007,Betzler2008,Bruckner2009,guillemotfinding,Koutis2012,Bjorklund2012}. The fastest (randomized) parameterized algorithm for \pegm runs in $\bigos(2^k)$ time for both colorful and multiset cases, and uses polynomial space~\cite{Bjorklund2012} (the $\bigos$ notation suppresses polynomial factors). Moreover, this last paper proves that a $O((2- \epsilon)^k)$-time algorithm is unlikely~\cite{Bjorklund2012}. Finally, the problem is unlikely to admit polynomial kernels, even on restricted classes of trees~\cite{Ambalath2010}. 

To deal with the high rate of noise in biological data, different variants of \pegm have been introduced. The approach of Dondi \textit{et al.} requires a solution with a minimum number of connected components~\cite{Dondi2011a}, while the one of Betzler \textit{et al.} asks for a 2-connected solution~\cite{Betzler2008}. As for traditional bioinformatics problems, some colors can be inserted in a solution, or conversely, some colors of the motif can be deleted in a solution~\cite{Bruckner2009,Dondi2011a,guillemotfinding}. Recently, Dondi \textit{et al.} introduced a variant when the number of substitutions between colors of the motif and colors in the solution must be minimum~\cite{Dondi2011}.

\paragraph{Our results}

Following this direction, we consider in Section~\ref{sec:maxmotif} an approximation issue when one wants to maximize the size of the solution. In Section~\ref{ref:subst}, we propose an inapproximability result when one wants to minimize the number of substitutions. Finally, we present in Section~\ref{sec:modules} a new requirement concerning the connectedness of the solution with one hardness result and two parameterized algorithms. 

\section{Maximizing the solution size}\label{sec:maxmotif}

To deal with the high rate of noise in the biological data, one approach allows some colors of the motif to be deleted from the solution, leading to \pmagm, a problem introduced by Dondi \textit{et al.}~\cite{Dondi2011a}.

\PbOpt{\pmagm}{A graph $G = (V,E)$, a set of colors $C$, a function $col : V \rightarrow C$, a multiset $M$ over $C$.}{A subset $V' \subseteq V$ such that (i) $G[V']$ is connected, (ii) $col(V') \subseteq M$ and such that $|V'|$ is maximized.} 


In a natural decision form, we are also given an integer $k$ in the input and one looks for a solution of size $k$ (the number of deletions is thus equal to $|M|-k$). The problem is known to be in the $\fpt$ class for parameter $k$~\cite{Dondi2011a,Bruckner2009,guillemotfinding}.
Concerning its approximation, \pmagm is $\apx$-hard, even when $G$ is a tree of maximum degree 3, the motif is colorful and each color appears at most twice in $G$ (in the same conditions, recall that the \pegm is polynomial~\cite{Dondi2011}). Moreover, there is no constant approximation ratio unless $\p = \np$, even when $G$ is a tree and $M$ is colorful~\cite{Dondi2011a}.

In the following, we answer an open question of Dondi \textit{et al.}~\cite{Dondi2011a} concerning the approximation issue of the problem when $G$ is a tree where each color occurs at most twice. More precisely, we prove that \pmagm cannot be approximated within $|V|^{\frac{1}{3}-\epsilon}, \forall \epsilon>0$, even when $G$ is a tree where each color appears at most twice and $M$ is colorful. 
To do so, we use a reduction from \pmais, a problem stated as follows: 
\PbOpt{\pmais}{A graph $G=(V,E)$.}{A subset $V' \subseteq V$ where there is no two nodes $u,v \in V'$ such that $\{u,v\} \in E$, and such that $|V'|$ is maximized.}

Our proof proceeds in four steps. We first describe the construction of the instance $\I'=(G,C,col)$ for \pmagm from the instance $\I=(G_I = (V_I,E_I))$ of \pmais (we consider the motif as $M = C$). We next prove that we can construct in polynomial time a solution for $\I'$ from a solution for $\I$ and, conversely, that we can construct in polynomial time a solution for $\I$ from a solution for $\I'$. Finally, we show that if there is an approximation algorithm with ratio $r$ for \pmagm, then there is an approximation algorithm with ratio $r$ for \pmais.

Before stating the reduction, consider a total order over the edges of $G$. We then define a function $\adj : V_I \to 2^{E_I}$, giving for a node $v \in V_I$, the ordered list of edges where $v$ is involved (thus of size $d(v)$, the degree of $v$). With this order, consider that $\adj(v)[i]$ give the $i$-th edge where $v$ is involved. From the graph  $G_I=(V_I,E_I)$, we build the graph $G=(V,E)$ as follows (see also Figure~\ref{fig:max motif}):

\begin{tabular}{ll}
-- $V =$ & $ \{r\}\ \cup$ $\{v_i^e : 1 \leq i \leq |V_I|, e \in \adj(v_i)\}\ \cup$\\
    & $\{v_i^j : 1 \leq i \leq |V_I|, 1 \leq j \leq |V_I|^2 \},$\\
\end{tabular}

\begin{tabular}{ll}
-- $E =$&$ \{ \{r,v_i^{\adj(v_i)[1]}\} : 1 \leq i \leq |V_I| \}\ \cup$\\
   & $\{\{v_i^{\adj(v_i)[j]},v_i^{\adj(v_i)[j+1]}\} : 1 \leq i \leq |V_I|, 1 \leq j < d(v_i) \}\ \cup$\\
   & $\{\{v_i^{\adj(v_i)[d(v_i)]},v_i^1\} : 1 \leq i \leq |V_I| \}\ \cup$ \\
   & $\{ \{v_i^j,v_i^{j+1} \} : 1 \leq i \leq |V_I|, 1 \leq j < |V_I|^2 \}$.\\
\end{tabular}

\begin{figure}[!ht]
 \begin{center}
\begin{tikzpicture}[scale=.7,transform shape,decoration={brace,amplitude=3pt},minimum size=25pt,inner sep=0pt]
	\node[] () at (-11,7.5) {$G_I=(V_I,E_I)$};
    \node[vertex,draw] (IS1) at (-9,9) {$v_1$};
    \node[vertex,draw] (IS2) at (-7,9) {$v_2$};
	\node[vertex,draw,very thick] (IS3) at (-8,10) {$v_3$};
    \node[vertex,draw,very thick] (IS4) at (-9,7.5) {$v_4$};
    \node[vertex,draw] (IS5) at (-7,7.5) {$v_5$};
     \path[draw] (IS1) -- (IS2) -- (IS3) -- (IS1);
     \path[draw] (IS2) -- (IS5) -- (IS4) -- (IS1);
	\begin{scope}[xshift=-9	cm]
	     \node[] () at (5,7.3) {$G=(V,E)$};
	     \node[vertex,draw,very thick] (Tr) at (10,13.8) {$c_r$};
  	     \node[vertex,draw] (T11) at (7,12.5) { $c_{\{1,2\}}$};
	     \node[vertex,draw] (T12) at (7,11.3) { $c_{\{1,3\}}$};
	     \node[vertex,draw] (T13) at (7,10.1) { $c_{\{1,4\}}$};
	     \path[draw] (Tr) -- (T11) -- (T12) -- (T13);
  	     \node[vertex,draw] (T21) at (8.5,12.5) { $c_{\{1,2\}}$};
	     \node[vertex,draw] (T22) at (8.5,11.3) { $c_{\{2,3\}}$};
	     \node[vertex,draw] (T23) at (8.5,10.1) { $c_{\{2,5\}}$};
	     \path[draw] (Tr) -- (T21) -- (T22) -- (T23);
  	     \node[vertex,draw,very thick] (T31) at (10,12.5) { $c_{\{1,3\}}$};
	     \node[vertex,draw,very thick] (T32) at (10,11.3) { $c_{\{2,3\}}$};
	     \path[draw] (Tr) -- (T31) -- (T32);
  	     \node[vertex,draw,very thick] (T41) at (11.5,12.5) { $c_{\{1,4\}}$};
	     \node[vertex,draw,very thick] (T42) at (11.5,11.3) { $c_{\{4,5\}}$};
	     \path[draw] (Tr) -- (T41) -- (T42);
  	     \node[vertex,draw] (T51) at (13,12.5) { $c_{\{2,5\}}$};
	     \node[vertex,draw] (T52) at (13,11.3) { $c_{\{4,5\}}$};
	     \path[draw] (Tr) -- (T51) -- (T52);
  	     \node[vertex,draw] (T1n1) at (7,8.9) {$c_1^1$};
  	     \node[vertex,draw] (T1n2) at (7,7.4) {$c_1^{25}$};
  	     \path[draw] (T13) -- (T1n1);
	     \path[draw,densely dotted] (T1n1) -- (T1n2);
    	 \node[vertex,draw] (T2n1) at (8.5,8.9) {$c_2^{1}$};
  	     \node[vertex,draw] (T2n2) at (8.5,7.4) {$c_2^{25}$};
  	     \path[draw] (T23) -- (T2n1);
	     \path[draw,densely dotted] (T2n1) -- (T2n2);
    	 \node[vertex,draw,very thick] (T3n1) at (10,8.9) {$c_3^{1}$};
  	     \node[vertex,draw,very thick] (T3n2) at (10,7.4) {$c_3^{25}$};
  	     \path[draw] (T32) -- (T3n1);
	     \path[draw,densely dotted] (T3n1) -- (T3n2);
    	 \node[vertex,draw,very thick] (T4n1) at (11.5,8.9) {$c_4^{1}$};
  	     \node[vertex,draw,very thick] (T4n2) at (11.5,7.4) {$c_4^{25}$};
  	     \path[draw] (T42) -- (T4n1);
	     \path[draw,densely dotted] (T4n1) -- (T4n2);
    	 \node[vertex,draw] (T5n1) at (13,8.9) {$c_5^{1}$};
  	     \node[vertex,draw] (T5n2) at (13,7.4) {$c_5^{25}$};
  	     \path[draw] 
    	 (T52) -- (T5n1);
	     \path[draw,densely dotted] 
    	 (T5n1) -- (T5n2);
    	 \draw[decorate,thick] (6.5,7) -- (6.5,9.3)
			node[midway,anchor=east,inner sep=2pt, outer sep=1pt]  {$5^2$}	;
\end{scope}
\end{tikzpicture}
 \end{center}
 \caption{Construction of $G$ from an instance $G_I$ of \pmais. For ease, only the color of the nodes of $G$ (not the label) are given. From a solution in $G_I$ in bold, the solution for \pmagm is given in bold in $G$.}\label{fig:max motif}
\end{figure}
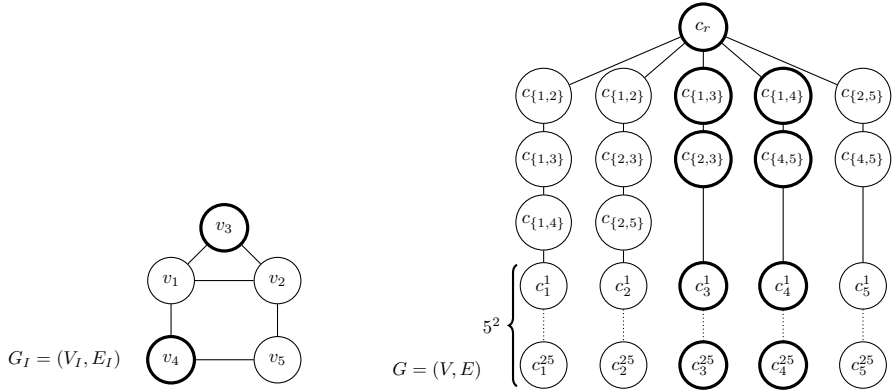

Informally speaking, $r$ is the root of $G$. There are $|V_I|$ paths connected to $r$. Each path represents a node of $G_I$ and is of length $d(v_i) + |V_I|^2$. Observe that $|V| = 1 + 2|E_I| + |V_I|.|V_I|^2$ (there are two nodes involved in each edge, therefore $\sum_{v \in V_I}d(v) = 2|E_I|$).
Let us now describe the set $C$ of colors and the coloration function $col : V \to C$. The set of colors is $C = \{c_r\} \cup \{c_e : e \in E_I\} \cup \{c_i^j : 1 \leq i \leq |V_I|, 1 \leq j \leq |V_I|^2\}$. Considering $M = C$, the motif is colorful. Coloration of the nodes of $G$ is done as follows: $col(r) = c_r$, $\forall e=\{v_i,v_j\} \in E_I, col(v_i^e) = col(v_j^e) = c_e$, $\forall 1 \leq i \leq |V_I|, 1 \leq j \leq |V_I|^2, col(v_i^j) = c_i^j$. In other words, for each edge $e = \{v_i,v_j\}$, a copy of the node $v_i$ and a copy of the node $v_j$ have the same color $c_e$, the one of the edge. Moreover, $r$ and the nodes $v_i^j$ all have different colors. In fact, nodes $\{v_i^j : 1 \leq i \leq |V_I|, 1 \leq j \leq |V_I|^2 \}$ can be considered as ``black boxes'', which are given for free. We clearly observe that by construction, $G$ is a tree where each color appears at most twice. Let us show how to build a solution for $\I'$ from a solution for $\I$. 

\begin{lemma}\label{lem:directMM}If there is a solution $V'_I \subseteq V_I$ for $\I$, then there is a solution $V' \subseteq V$ for $\I'$ such that $|V'| \geq |V'_I|.|V_I|^2$.
\end{lemma}
\begin{proof}
We build $V'$ as follows: $V' = \{r\} \cup \{v_i^e, v_i^j : v_i \in V'_I, e \in \adj(v_i), 1 \leq j \leq |V_I|^2\}$. In other words, we add in $V'$ the root $r$ of the tree and all the paths corresponding to the nodes of $V'_I$ (see also Figure~\ref{fig:max motif}). 

Let us prove that $V'$ is a solution for $\I'$ such that $|V'| \geq |V'_I|.|V_I|^2$. Since the root is in the solution, $G[V']$ is connected. Moreover, colors of $V'$ are all distinct, therefore the solution is colorful. Indeed, if there are $u,v \in V'$ such that $col(u) = col(v)$, then $\{u,v\} \in E$, which is a contradiction since $V'_I$ is a solution for \pmais. Finally, we bound the size of $V'$ by observing that for each $v \in V'_I$, we add the set of nodes in the path corresponding to $v$, which is of size $d(v) + |V_I|^2 \geq |V_I|^2$.\qed
\end{proof}

Let us now show how to build a solution for $\I$ from a solution for $\I'$.

\begin{lemma}\label{lem:retourMM}If there is a solution $V' \subseteq V$ for $\I'$, then there is a solution $V'_I \subseteq V_I$ for $\I$ such that $|V'_I| \geq \left \lceil \frac{|V'| - 2|E_I| - 1}{|V_I|^2} \right \rceil$.
\end{lemma}
\begin{proof}
For each $1 \leq i \leq |V_I|$, we add $v_i$ in $V'_I$ iff all the nodes $v_i^j, 1 \leq j \leq |V_I|^2$ and $v_i^e, e \in \adj(v_i)$ are in $V'$. In other words, we add $v_i$ in $V'_I$ if the whole path corresponding to this node is in $V'$. 

Let us prove that $V'_I$ is a solution for $\I$ such that $|V'_I| \geq \left \lceil \frac{|V'| - 2|E_I| - 1}{|V_I|^2} \right \rceil$. If there are $v_i,v_j \in V'_I$ such that $\{v_i,v_j\} = e \in E$, then $v_i^e$ and $v_j^e$ are in $V'$. It is impossible since $col(v_i^e) = col(v_j^e)$ and since all the colors of $V'$ must be distinct to be a solution for \pmagm. Consequently, $V'_I$ is an independent set. There are $\left \lceil \frac{|V'| - 2|E_I| - 1}{|V_I|^2} \right \rceil$ whole paths in $V'$. Indeed, by removing $2|E_G| + 1$ to the whole number of nodes in the solution, we bound the number of nodes of type $v_i^j$ (recall that $|V| = 1 + 2|E_I| + |V_I|.|V_I|^2$).\qed 
\end{proof}

These two lemmas lead to the main result of this section.

\begin{proposition}\label{prop:approx magm} Unless $\p = \np$, \pmagm cannot be approximated within $|V|^{\frac{1}{3} - \epsilon}$ in polynomial time, for any $\epsilon > 0$, even when the motif is colorful and $G$ is a tree where each color of $C$ appears at most twice.
\end{proposition}
\begin{proof}


Suppose there is such a ratio $r$ for \pmagm. Then, there is an approximate solution $V'_{APX}$ which, compared to the optimal solution $V'_{OPT}$, is of size $|V'_{APX}| \geq \frac{|V'_{OPT}|}{r}$. 

\noindent
With Lemma~\ref{lem:directMM}, $  |V'_{OPT}|   \geq |V'_{I_{OPT}}|.|V_{I}|^2.$\\
We supposed $   |V'_{APX}| \geq  \frac{|V'_{OPT}|}{r}.$\\
Therefore, $  |V'_{APX}|   \geq \frac{|V'_{I_{OPT}}|.|V_{I}|^2}{r}.$\\
With Lemma~\ref{lem:retourMM}, $ |V'_{I_{APX}}|  \geq \left \lceil \frac{|V'_{APX}| - 2|E_I| - 1}{|V_I|^2} \right \rceil.$ \\
Which leads to, $ |V'_{I_{APX}}|  \geq \left \lceil \frac{((|V'_{I_{OPT}}|.|V_{I}|^2)/r) - 2|E_I| - 1}{|V_I|^2} \right \rceil.$ \\
Since $\frac{2|E_I| + 1}{|V_I|^2} \leq 1$, $ |V'_{I_{APX}}|  \geq \left \lceil \frac{(|V'_{I_{OPT}}|.|V_{I}|^2)/r}{|V_I|^2} \right \rceil - 1$.\\ 
Finally, $ |V'_{I_{APX}}|  \geq \frac{|V'_{I_{OPT}}|}{r} - 1.$

Thereby, if there is an approximation algorithm with ratio $r$ for \pmagm, there is an approximation algorithm with ratio $r$ for \pmais. We conclude the proof by observing that $|V| = \bigo(|V_I|^3)$ and that unless $\p = \np$, \pmais cannot be approximated in polynomial time within $|V_I|^{1 - \epsilon}$, $\forall \epsilon > 0$~\cite{Zuckerman2007}.\qed
\end{proof}

\section{Minimizing the number of substitutions}\label{ref:subst}

In this section, we focus on \pmisgm, a problem recently introduced by Dondi \textit{et al.}~\cite{Dondi2011}. In this variant, some colors of the motif can be deleted, but the size of the solution must be equal to $|M|$. Therefore, the deleted colors must be substituted by the same number of colors. 




\PbOpt{\pmisgm}{A graph $G = (V,E)$, a set of colors $C$, a function $col : V \rightarrow C$, a multiset $M$ over $C$.}{A subset $V' \subseteq V$ such that (i) $|V'|=|M|$, (ii) $G|V']$ is connected and such that the number of substitutions to get $M$ from $col(V')$ is minimized.} 



Dondi \textit{et al.}~\cite{Dondi2011} prove that \pmisgm is $\np$-hard, even when $G$ is a tree of maximum degree $4$ where each color occurs at most twice and the motif is colorful. On the positive side, they prove that the decision version of the problem is in the $\fpt$ class 
when the parameter is the size of the solution. Another algorithm is provided by Koutis in time $\bigos(5.08^k)$ and polynomial space~\cite{Koutis2012}. We however remark that we can use the algorithm of~\cite{guillemotfinding} for \textsc{List Colored Graph Motif} to solve \pmisgm in $\bigos(4^k)$, by introducing $p$ new colors, and adding them in the list of colors of each node of the graph. Therefore, we can still look for a solution of size $k$, with at most $p$ substitutions. 

Let us remark that for instances where the optimum equals 0, then the problem is not approximable at all. Indeed, any polynomial-time approximation algorithm for \pmisgm would give a solution in polynomial time to the $\np$-complete \pegm problem. However, for instances where optimum contains at least one substitution, we prove that even in restrictive conditions (when $G$ is a rooted tree of depth 2 and the motif is colorful), there is no polynomial-time approximation algorithm with ratio $c \log |V|$, for a constant $c$.

To prove such inapproximability result, we propose an \mbox{E-reduction} from \pmisc~\cite{Raz1997}, a problem stated as follows: 
\PbOpt{\pmisc}{A set $X = \{x_1,x_2,\dots,x_{|X|}\}$, a collection $\S = \{S_1,S_2,\dots,S_{|\S|}\}$ of subsets of $X$}{A subset $\T \subseteq \S$ such that every element of $X$ belongs to at least one member of $\T$, and such that $|\T|$ is minimized.}

We denote by $e(i,j)$ the index $l$ such that $x_l$ correspond to the $j$-th element of $S_i$. 
We first describe the polynomial construction of $\I'=(G,C,col,M)$, instance of \pmisgm, from $\I = (X,\S)$, any instance of \pmisc. 
From an instance $\I$, let build $G=(V,E)$ as follows (see also Figure~\ref{fig:min subst}):

\begin{tabular}{ll}
-- $V =$ & $\{r\}\ \cup $ $ \{v_i : 1 \leq i \leq |\S| \}\ \cup$\\
      & $ \{v_{i,j,t} : 1 \leq i \leq |\S|, 1 \leq j \leq |S_i|, 1 \leq t \leq |\S|+1\}$, 
\end{tabular}

\begin{tabular}{ll}
-- $E =$ & $\{ \{r,v_i\} : 1 \leq i \leq |\S| \}\ \cup$ \\
        & $ \{\{v_i, v_{i,j,t}\} : 1 \leq i \leq |\S|, 1 \leq j \leq |S_i|, 1 \leq t \leq |\S|+1 \}$. 
\end{tabular}

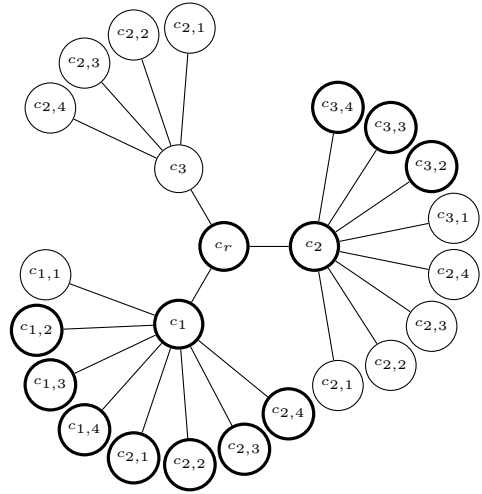
\begin{figure}[h!]
\centering
\begin{tikzpicture}[scale=.85, minimum size=18,grow cyclic]
	\tikzstyle{level 1} = [sibling distance=6.5cm,level distance = 1.4cm,sibling angle=120]
	\tikzstyle{level 2} = [sibling distance=1.5cm,level distance = 2.2cm,,sibling angle=23]
\begin{tiny}
	\node[vertex,very thick] (r) {$c_r$}
		child {node[vertex,very thick] (S1) {$c_1$}
			child {node[vertex] (e111) {$c_{1,1}$}}
			child {node[vertex,very thick] (e112) {$c_{1,2}$}}
			child {node[vertex,very thick] (e113) {$c_{1,3}$}}
			child {node[vertex,very thick] (e114) {$c_{1,4}$}}
			child {node[vertex,very thick] (e121) {$c_{2,1}$}}
			child {node[vertex,very thick] (e122) {$c_{2,2}$}}
			child {node[vertex,very thick] (e123) {$c_{2,3}$}}
			child {node[vertex,very thick] (e124) {$c_{2,4}$}}
		}
		child {node[vertex,very thick] (S2) {$c_2$}
			child {node[vertex] (e211) {$c_{2,1}$}}
			child {node[vertex] (e212) {$c_{2,2}$}}
			child {node[vertex] (e213) {$c_{2,3}$}}
			child {node[vertex] (e214) {$c_{2,4}$}}
			child {node[vertex] (e221) {$c_{3,1}$}}
			child {node[vertex,very thick] (e222) {$c_{3,2}$}}
			child {node[vertex,very thick] (e223) {$c_{3,3}$}}
			child {node[vertex,very thick] (e224) {$c_{3,4}$}}
		}
		child {node[vertex] (S3) {$c_3$}
			child {node[vertex] (e311) {$c_{2,1}$}}
			child {node[vertex] (e312) {$c_{2,2}$}}
			child {node[vertex] (e313) {$c_{2,3}$}}
			child {node[vertex] (e314) {$c_{2,4}$}}
		}	
	;
\end{tiny}
%
%
%



\end{tikzpicture}
\caption{Illustration of the construction of an instance of \pmisgm from an instance of \pmisc such that $X = \{x_1, x_2, x_3\}$ and $\S = \{\{x_1,x_2\}, \{x_2,x_3\}, \{x_2\}\}$. For ease, only the color of each node of the graph (and not the label) is given. The associated motif is $M = \{c_r\} \cup \{c_{k,t} : 1 \leq k \leq 3, 1 \leq t \leq 4\}$. A possible solution (with two substitutions) is given in bold.}\label{fig:min subst}
\end{figure}

Informally speaking, $r$ is the root of a tree with $|\S|$ children, corresponding to each subset of $\S$. Each child $v_i, 1 \leq i \leq |\S|$, got $(|\S|+1)|S_i|$ children, corresponding to $|\S|+1$ copies of each element of $S_i$. 
The set of colors is $C = \{c_r\} \cup \{c_i : 1 \leq i \leq |\S| \} \cup \{c_{k,t} : 1 \leq k \leq |X|, 1 \leq t \leq |\S|+1\} 
$. The coloring function is such that the root has a unique color, \textit{i.e.} $col(r) = c_r$. Each node $v_i$ is colored with the unique color corresponding to the subset of $\S$, $col(v_i) = c_i, \forall 1 \leq i \leq |\S|$. Each node $v_{i,j,t}$  get the color of the copy of the represented element, \textit{i.e.} $col(v_{i,j,t}) = c_{e(i,j),t}$. 
Finally, the motif is $M = \{c_r\} \cup \{c_{k,t} : 1 \leq k \leq |X|,1 \leq t \leq |\S|+1\}$. Observe that the colors $\{c_i : 1 \leq i \leq |\S|\}$ 
are not in the motif (which is colorful by construction). Let now show how to build a solution for $\I'$ from a solution for $\I$, and \textit{vice-versa}.

\begin{lemma}\label{lem:direct}If there is a solution $\T$ for an instance $\I$ of \pmisc, there is a solution for the instance $\I'$ of \pmisgm with $|\T|$ substitutions.
\end{lemma}
\begin{proof}Let $\T \subseteq \S$ be a solution for $\I$. Given a total order on the subsets of $\S$, for each $1 \leq k \leq |X|$, denote by $S_{min}^k$ the subset such that (i) $S_{min}^k \in \T$, and (ii) $S_{min}^k$ is the first subset of $\T$ where $x_k$ is. Moreover, for each $S_i$, denote by $f_i$ the smallest index $j$ of $v_{i,j,t}$ such that $S_i = S_{min}^{e(i,j)}$. 

The solution $V'$ is  built as follows: $V' = \{r\} \cup \{v_i : S_i \in \T\} \cup \{v_{i,j,t} : S_i = S_{min}^{e(i,j)}, j=f_i, 2 \leq t \leq |\S|+1 \} \cup \{v_{i,j,t} : S_i = S_{min}^{e(i,j)}, j \neq f_i, 1 \leq t \leq |\S|+1 \}$. Less formally, we put in the solution the root, the set of nodes representing subsets $S_i$ of $\T$, also with the $|\S|+1$ copies of each node representing an $x_k$ (the one in the subset with minimal index in the solution), except for the element $x_k$ of $X$ with the lower index in $S_{min}^k$, where only $|\S|$ copies are in the solution.

The graph $G[V']$ is connected since the nodes $v_{i,j,t}$ are in the solution if and only if the node $v_i$ is also in the solution. Moreover, a node $v_i$ is in the solution if there is a $k$ such that $S_i = S_{min}^k$. There is thus an integer $f_i$ for which only $|\S|$ copies of $v_{i,f_i,t}$ are in the solution. Therefore, by construction, the color $c_{e(i,f_i),1}$, which is in the motif, is substituted in the solution by $c_i$. On the whole, there are $|\T|$ substitutions, since the other colors of the motif are in the solution.\qed
\end{proof}

\begin{lemma}\label{lem:retour}From a solution for the instance $\I'$ for \pmisgm with at most $s$ substitutions, there is a solution for the instance $\I$ for \pmisc of size at most $s$.
\end{lemma}

\begin{proof}
Let $V' \subseteq V$ be a solution for $\I'$ such that we can obtain $M$ from $col(V')$ with at most $s$ substitutions. We can suppose that $s < |\S|+1$, otherwise, $\T = \S$ is a solution of correct size.
Solution for $\I$ is built as follows: $\T = \{S_i : v_i \in V'\}$. 

If for some $1 \leq k \leq |X|$ there is no color of the set $\{c_{k,t} : 1 \leq t \leq |\S|+1\}$ in the solution, it means that these $|\S|+1$ colors have all been substituted, which is a contradiction with the supposed maximum number of $s$ substitutions. Therefore, for each $1 \leq k \leq |X|$, there is at least one color from the set $\{c_{k,t} : 1 \leq t \leq |\S|+1\}$ in the solution. By definition, a solution must be connected, therefore, all elements of $X$ are covered by $\T$. Finally, the size of $\T$ is bounded by $s$. Indeed, since their colors are not in the motif, there are at most $s$ nodes $v_i$ in $V'$.\qed

%
%
\end{proof}

The above construction and two Lemmas lead to the result concerning the approximation of \pmisgm.

\begin{proposition}\label{prop:min subst}
Unless $\p = \np$, \pmisgm cannot be approximated in polynomial time within $c \log |V|$, where $c$ is a constant, even when the motif is colorful and $G$ is a rooted tree of depth 2. 
\end{proposition}
\begin{proof}
The proof comes directly from the Lemmas~\ref{lem:direct} and \ref{lem:retour}, and because \pmisc cannot be approximated in polynomial time within $c \log |X|$, unless $\p = \np$~\cite{Raz1997}. Observe that from any solution $V'$ for $\I'$, we can construct in polynomial time a solution $\T$ for $\I$ such that $val(\I,\T) = val(\I', V')$, thus $opt(\I) \leq opt(\I')$. A consequence of Lemma~\ref{lem:direct} is that $opt(\I') \leq opt(\I)$. We then have $opt(\I') = opt(\I)$ and 
consequently, $\varepsilon(\I,\T) = \frac{val(\I,\T)}{opt(\I)} - 1  \leq \frac{val(\I',V')}{opt(\I')} - 1 = \varepsilon(\I',V')$. 
\qed
\end{proof}



\section{Using modularity}\label{sec:modules}

In this section, we introduce a variant of \pegm, where the connectivity constraint is replaced by modularity. After a quick recall on the modules properties, we justify this new variant. The problem remains $\np$-hard, however, the tools offered by the modularity allow efficient algorithms.

\subsection{Definitions and properties}

In an undirected graph $G=(V,E)$, a node $x$ \textit{separates} two nodes $u$ and $v$ iff $\{x,u\} \in E$ and $\{x,v\} \notin E$. A \textit{module} $\M$ of a graph $G$ is a set of nodes not separated by any node of $V \setminus \M$. In other words, a module $\M$ is such that $\forall x \notin \M, \forall u,v \in \M, \{x,u\} \in E \Leftrightarrow \{x,v\} \in E$~\cite{Chein1981} (see also Figure~\ref{fig:example modules}). 
The whole set of nodes $V$ and any singleton set $\{u\}$, where $u \in V$, are the trivial modules.

\begin{figure}[ht!]
\begin{center}
\begin{tikzpicture}[scale=.55,>=stealth,shorten <=.5pt,shorten >=.5pt]
	\node[vertex,draw] (1) at (2.5,0.5) {$v_1$};
	\node[vertex,draw] (2) at (4,0.5) {$v_2$};
	\node[vertex,draw] (3) at (4,2) {$v_3$};
	\node[vertex,draw] (4) at (4,-1) {$v_4$};
	\node[vertex,draw] (5) at (6,0.5) {$v_5$};
	\node[vertex,draw] (6) at (6,-1) {$v_6$};
	\node[vertex,draw] (7) at (8,0.5) {$v_7$};
	\node[vertex,draw] (8) at (8,-1) {$v_8$};
	\path[edge] 
      		(1) -- (3) -- (5) -- (4) -- (1);
	\path[edge] 
      		(2) -- (3) -- (6) -- (2);
	\path[edge] 
      		(2) -- (4) -- (6);
	\path[edge] 
		(2) -- (5);
	\path[edge] 
		(7) -- (8);
\end{tikzpicture}
\end{center}
\caption{Example of modules : the whole set of nodes $V$, the singletons $\{v\}, \forall v \in V$, the two connected component $\{v_1,v_2,v_3,v_4,v_5,v_6\}, \{v_7,v_8\}$, or the sets $\{v_3,v_4\}, \{v_5,v_6\}, \{v_1,v_2,v_5,v_6\}$.}\label{fig:example modules}
\end{figure}
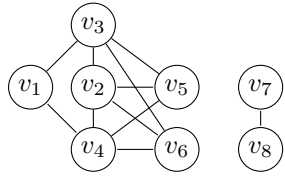

Before stating the definition of specific modules, let say that two modules $A$ and $B$  \textit{overlap} if (i) $A \cap B \neq \emptyset$, (ii) $A \setminus B \neq \emptyset$, and (iii) $B \setminus A \neq \emptyset$. According to~\cite{Chein1981}, if two modules $A$ and $B$ overlap, then $A \cap B$, $A \cup B$ and $(A \cup B) \setminus (A \cap B)$ are also modules. This allows the definition of \textit{strong modules}. A module is \textit{strong} if no other module overlaps it, otherwise it is \textit{weak}. Therefore, two strong modules are either included into the other, either of empty intersection. A module $\M \subset S$ is said \textit{maximal} for a given set of nodes $S$ (by default the set of nodes $V$) if there is no module $\M'$ s.t. $\M \subset \M' \subset S$. In other words, the only module which contains the maximal module $\M$ is $S$. 

There are three types of modules : (i) \textit{parallel}, when the subgraph induced by the nodes of the module is not connected (it is a parallel composition of its connected components), (ii) \textit{series}, when the complement of the subgraph induced by the nodes of the module is not connected (it is a series composition of the connected components of its complement), or (iii) \textit{prime}, when both the subgraph induced by the nodes of the module and its complement are connected.

The inclusion order of the maximal strong modules defines the \textit{modular tree decomposition} $\T(G)$ of $G$, which is enough to store the whole set of strong modules. The tree $\T(G)$ can be recursively built by a top-down approach, where the algorithm recurs on the graph induced by the considered strong module. The root of this tree is the set of all nodes $V$ while the leaves are the singleton sets $\{u\}$, $\forall u \in V$. Each node of $\T(G)$ got a label representing the type of the strong module, \textit{parallel}, \textit{series} or \textit{prime}. Children of an internal node $\M$ are the maximal submodules of $\M$ (\textit{i.e.} they are disjoints). Figure \ref{fig:example tree} gives an example of the construction of $\T(G)$ from a sample graph $G$.
 The modular tree decomposition can be obtained with a linear time algorithm, (\textit{e.g.} the one described in \cite{Habib2004}).
We can now introduce an essential property of $\T(G)$:

\begin{theorem} (\cite{Chein1981})
A module of $G$ is either a node of $\T(G)$, either a union of children (of depth 1) of a series or parallel node in $\T(G)$.  
\end{theorem}

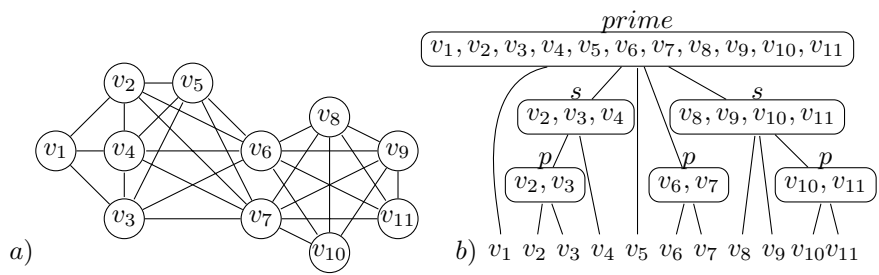
\begin{figure}[ht!]
\begin{center}
\begin{tikzpicture}[scale=.45,>=stealth,shorten <=.5pt,shorten >=.5pt]
	\node[] () at (-1,-3) {$a)$};
	\node[bigvertex,draw] (1) at (0,0) {$v_1$};
	\node[bigvertex,draw] (2) at (2,2) {$v_2$};
	\node[bigvertex,draw] (3) at (2,-2) {$v_3$};
	\node[bigvertex,draw] (4) at (2,0) {$v_4$};
	\node[bigvertex,draw] (5) at (4,2) {$v_5$};
	\node[bigvertex,draw] (6) at (6,0) {$v_6$};
	\node[bigvertex,draw] (7) at (6,-2) {$v_7$};
	\node[bigvertex,draw] (8) at (8,1) {$v_8$};
	\node[bigvertex,draw] (9) at (10,0) {$v_9$};
	\node[bigvertex,draw] (10) at (8,-3) {$v_{10}$};
	\node[bigvertex,draw] (11) at (10,-2) {$v_{11}$};
	\path[edge] 
      		(1) -- (2) -- (4) -- (1) -- (3) -- (4);
	\path[edge] 
      		(2) -- (5) -- (4) -- (7) -- (5) -- (3);
	\path[edge] 
      		(2) -- (6) -- (5);
	\path[edge] 
		(2) -- (7) -- (3) -- (6) -- (4);
	\path[edge] 
		(6) -- (8) -- (9) -- (11) -- (6);
	\path[edge]
		(6) -- (9) -- (10) -- (6);
	\path[edge]
		(7) -- (9);
	\path[edge]
		(7) -- (10) -- (8) -- (7) -- (11) -- (8);
\begin{scope}[xshift=1cm]
	\node[] () at (11,-3) {$b)$};
	\node[] () at (16,3.8) {$prime$};
	\node[rect,draw] (tV) at (16,3)
		{$v_1,v_2,v_3,v_4,v_5,v_6,v_7,v_8,v_9,v_{10},v_{11}$};
	\node[] () at (14.2,1.7) {$s$};
	\node[rect,draw] (t234) at (14.2,1) {$v_2,v_3,v_4$};
	\node[] () at (19.5,1.7) {$s$};
	\node[rect,draw] (t891011) at (19.5,1) {$v_8,v_9,v_{10},v_{11}$};
	\node[] () at (13.3,-0.3) {$p$};
	\node[rect,draw] (t23) at (13.3,-1) {$v_2,v_3$};
	\node[] () at (17.5,-0.3) {$p$};
	\node[rect,draw] (t67) at (17.5,-1) {$v_6,v_7$};
	\node[] () at (21.5,-0.3) {$p$};
	\node[rect,draw] (t1011) at (21.5,-1) {$v_{10},v_{11}$};
	\node[rect] (t1) at (12,-3) 
		{$v_{1}$};
	\node[rect] (t2) at (13,-3) 
			{$v_{2}$};
	\node[rect] (t3) at (14,-3) 
			{$v_{3}$};
	\node[rect] (t4) at (15,-3) 
			{$v_{4}$};
	\node[rect] (t5) at (16,-3) 
			{$v_{5}$};
	\node[rect] (t6) at (17,-3) 
			{$v_{6}$};
	\node[rect] (t7) at (18,-3) 
			{$v_{7}$};
	\node[rect] (t8) at (19,-3) 
			{$v_{8}$};
	\node[rect] (t9) at (20,-3) 
			{$v_{9}$};
	\node[rect] (t10) at (21,-3) 
			{$v_{10}$};
	\node[rect] (t11) at (22,-3) 
			{$v_{11}$};
	\path[edge]
		(tV) .. controls +(-5,-1) and +(0,1)  .. (t1);
	\path[edge]
		(tV) -- (t234) -- (t23) -- (t2);
	\path[edge]
		(t23) -- (t3);
	\path[edge]
		(t234) -- (t4);
	\path[edge]
		(tV) -- (t5);
	\path[edge]
		(tV) -- (t67) -- (t6);
	\path[edge]
		(t67) -- (t7);
	\path[edge]
		(tV) -- (t891011) -- (t8);
	\path[edge]
		(t891011) -- (t9);
	\path[edge]
		(t891011) -- (t1011) -- (t10);
	\path[edge]
		(t1011) -- (t11);
\end{scope}
\end{tikzpicture}
\end{center}
\caption{A sample graph in a) and the corresponding modular tree decomposition in b). Nodes of the tree are either series ($s$), parallel ($p$), prime ($prime$) or leaves.
}\label{fig:example tree}
\end{figure}


One can see strong modules as generators of the modules of $G$: the set of all modules of $G$ can be obtained from the tree $\T(G)$. A crucial point to note is that there is potentially an exponential number of modules in a graph (\textit{e.g.}, the clique $K_n$ has $2^n$ modules), but the size of $\T(G)$ is $\bigo(n)$ (more precisely, $\T(G)$ has less than $2n$ nodes since there are $n$ leaves and no node with exactly one child). Therefore, the exponential-sized family of modules of $G$ can be represented by the linear sized tree $\T(G)$.

\subsection{When modules join \GM}

In the following, we investigate the algorithmic issues of other topology-free definition, when replacing the connectedness demand by modularity. Following definition of \pegm, we introduce \pmogm.

\PbDec{\pmogm}{A graph $G = (V,E)$, a set of colors $C$, a function $col : V \rightarrow C$, a multiset $M$ on $C$ of size $k$.}{Does there exist a subset $V' \subseteq V$ such that (i) $V'$ is a module of $G$ and (ii) $col(V')=M$.}

This definition links the modularity demand with the motif research. The module definition implies that all the nodes in this module have a uniform relation with the set of all the other nodes outside of the module. The module nodes are indistinguishable from the outside, they are acting similarly with the other nodes of the graph.

Authors of \cite{Alm2003} define a biological module as a set of elements having a separable function from the rest of the graph. Similarly, authors of~\cite{Ravasz2002} describe a biological module as a set of some elements with an identifiable task, separable from the functions of the other biological modules. Moreover, it is shown in~\cite{Constanzo2010} that genes with a similar neighborhood have chances to be in a same biological process. It is thus possible that set of nodes in an algorithmic module of a graph representing a biological network have a common biological function. Also note that authors of~\cite{Segal2003} describe modules in gene regulatory networks as groups of genes which obey to the same regulations, and consequently, as groups which members cannot be distinguished from the rest of the network.

Moreover, apart of using modules in a slightly different goal (in order to predict more cleverly results of PPI), Gagneur \textit{et al.}~\cite{Gagneur2004} note that modules of a graph can join biological modules, and consider modular decomposition as a general tool for biological network analysis under different representations (oriented graphs, hyper-graphs...).

However, there is no clear definition of what is (or should be) a biological module in a network~\cite{Alm2003}. We thus claim that the approach using modular decomposition is complementary to the previous definitions of biological modules (\textit{e.g.} connected occurrences or compact occurrences).

\subsection{Difficulty of the problem}

Unfortunately, \pmogm is $\np$-hard, even under strong restrictions, \textit{i.e.} when $G$ is a collection of paths of size three, and when the motif is colorful. Observe that under the same conditions, \pegm is trivially polynomial-time solvable.

\begin{proposition} \pmogm is $\np$-complete even if $G$ is a collection of paths of size 3 and $M$ is colorful.
\end{proposition}
\begin{proof} \pmogm is in $\np$ since given a set $V' \subseteq V$, one can check in polynomial-time if $V'$ is a module and if the colors of $C$ appears exactly once if $V'$. To prove its hardness, we propose a reduction from \textsc{Exact Cover by 3-Sets} (\pxtc). This special case of \textsc{Set Cover} is known to be $\np$-complete. Recall that \pxtc is stated as follows: 
\PbDec{\pxtc}{A set $X = \{x_1,x_2,\dots,x_{3q}\}$ and a collection $\S = \{S_1,\dots,S_{|\S|}\}$ of 3-elements subsets of $X$.}{Does $\S$ contains a subcollection $\T \subseteq S$ such that each element of $X$ occurs in exactly one element of $T$.}
Size of $X$ must be a multiple of three since a solution is a set of triplets where each element of $X$ must appears exactly once.

Let us now describe the construction of an instance $\I'=(G,C,col)$ of \pmogm from an arbitrary instance $\I = (X,\S)$ of \pxtc (see also Figure~\ref{fig:X3Cconstruction}). The graph $G=(V,E)$ is built as follows: $V = \{ v_i^j : 1 \leq i \leq |\S|, x_j \in S_i \}$,  $E = \{ \{v_i^1,v_i^2\} \cup \{v_i^2,v_i^3\} : 1 \leq i \leq |\S| \}$. 	Informally speaking, $G$ is a collection of $|\S|$ paths with three nodes (recall that for each $1 \leq i \leq |\S|, |S_i| = 3$).

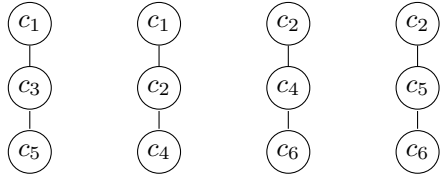
\begin{figure}[ht!]
\begin{center}
\begin{tikzpicture}[scale=.85,>=stealth,shorten <=.5pt,shorten >=.5pt]
	\node[vertex,draw] (11) at (0,0) {$c_1$};
	\node[vertex,draw] (13) at (0,-1) {$c_3$};
	\node[vertex,draw] (15) at (0,-2) {$c_5$};
	\node[vertex,draw] (21) at (2,0) {$c_1$};
	\node[vertex,draw] (22) at (2,-1) {$c_2$};
	\node[vertex,draw] (24) at (2,-2) {$c_4$};
	\node[vertex,draw] (32) at (4,0) {$c_2$};
	\node[vertex,draw] (34) at (4,-1) {$c_4$};
	\node[vertex,draw] (36) at (4,-2) {$c_6$};
	\node[vertex,draw] (42) at (6,0) {$c_2$};
	\node[vertex,draw] (45) at (6,-1) {$c_5$};
	\node[vertex,draw] (46) at (6,-2) {$c_6$};
	
	\path[draw] 
      		(11) -- (13) -- (15);
	\path[draw] 
      		(21) -- (22) -- (24);
	\path[draw] 
      		(32) -- (34) -- (36);
	\path[draw] 
      		(42) -- (45) -- (46);
\end{tikzpicture}
\end{center}
\caption{The graph $G$ built from $X = \{x_1,x_2, \dots ,x_6\}$ (thus with $q=2$) and $\S = \{\{x_1,x_3,x_5\},\{x_1,x_2,x_4\},\{x_2,x_4,x_6\},\{x_2,x_5,x_6\}\}$ (only the colors of the node are written). By construction, the set of colors asked in any solution is $C = \{c_1,c_2, \dots ,c_6\}$.}\label{fig:X3Cconstruction}
\end{figure}

The set of colors is $C = \{c_i : 1 \leq i \leq |X|\}$. The coloration of $G$ is such that $col(v_i^j) = c_j$. In other words, each node get the color of the represented element of $X$. We also consider the colorful motif as $M=C$. This construction is clearly done in polynomial-time in regards of $\I$.

Let us now prove that if there is a solution for an instance $\I$ of \pxtc, then there is solution for the instance $\I'$ of \pmogm. Given a solution $\T \subseteq \S$ for $\I$, a solution $V'$ for $\I'$ is built as follows: $V' = \{v_i^j : S_i \in \T, x_j \in S_i \}$. Informally speaking, the solution contains the set of paths corresponding to the chosen triplets in the solution for \pxtc. The set $V'$ is a module, and by definition of a solution for $\I$, each color of $\{c_i : 1 \leq i \leq |X|\}$ appears exactly once in $V'$. 

Conversely, let us now prove that there is a solution for the instance $\I$ of \pxtc if there is a solution for the instance $\I'$ of \pmogm. First observe that since $q \geq 1$, then $|X| \geq 3$ and therefore $C \geq 3$. A module of size greater or equal than three in a collection of paths of size three must be a union of paths of size three. Indeed, suppose by contradiction that there is a module $\M$ of size greater than three which is not a union of paths of size three. There is thus a node $u \in \M$ such that at least one of its neighbor $v \in N(u)$ is not in $\M$ and $v$ separates $u$ from another node of $\M$. Therefore, $\M$ is not a module. The solution is built as follows: $\T = \{S_i : v_i^j \in V' \}$. Since the solution $V'$ is a union of paths of size three, each triplet $S_i$ is either completely chosen in the solution $\T$, either absent. Moreover, since $V'$ is a solution, colors of $V'$ appears exactly once. Therefore, each element of $X$ appears exactly once in $\T$.\qed
\end{proof}

\subsection{Algorithms for the decision problem}

Even if the problem is hard under strong restrictions, the modular decomposition tree is a useful structure to design efficient algorithms. More precisely, we show in the sequel that \pmogm is in the $\fpt$ class when the parameter is the size of the solution. 
As a corollary, we show that the problem can be solved in polynomial time if the number of colors is bounded, or more generally, that \pmogm is in $\fpt$ when parameterized by $(k,|C|)$. Moreover, \pmogm is still in the $\fpt$ class if a set of colors is associated to each node of the graph.

Let us first observe that asking for a strong module instead of any module in the definition of \pmogm leads to a trivial linear algorithm. Indeed, one can just browse $\T(G)$ and test if the set of colors for each strong module is equal to the motif.

Let us now show an algorithm with a time complexity of $\bigos(2^k)$, where $k$ is the size of the solution, for \pmogm, even if the motif is a multiset. 

\begin{proposition}\label{prop:mmogm}
There is a parameterized algorithm for \pmogm with a time complexity of $\bigo(2^k|V|^2)$ and a space complexity of $\bigo(2^k|V|)$, where $k$ is the size of the motif and of the solution.
\end{proposition}
\begin{proof}
Since $|M|=k$, observe that there are at most $2^k$ different multisets $M'$ such that $M' \subseteq M$. 
We first build in polynomial time the modular tree decomposition $\T(G)$ from $G$. We repeat the following algorithm for each node $\M$ of $\T(G)$.

We start by testing if the set of the colors of $\M$ is exactly equal to the motif $M$. If it is the case, the algorithm terminates. Otherwise, if $\M$ is a series or parallel node, a module can be a union of its children. Given an arbitrary order on its $t$ children, denote by $\child(\M)[i]$ the $i$-th child of $\M$. We then delete all children $\M'$ of $\M$ such that $col(\M') \not \subset M$, where $col(\M')$ is the set of colors of the nodes of $\M'$. Indeed, such a child cannot be in a solution considering $M$. We note that the set of colors for each child correspond to a multiset $M' \subseteq M$. 

Since any union of children of $\M$ is a module of $G$, it is thus a potential solution. We propose to test by dynamic programming if such union corresponds to a solution for $M$. We build a table $D(i,M')$, for $0 \leq i \leq t$ and $M' \subseteq M$. Therefore, $D$ has $t+1$ lines and $2^k$ columns. We fill this table as follows:
\begin{align*}
D(0,M') &= True  \text{ if } M' = \{0, \dots ,0\} \text{, $False$ otherwise,}\\
D(i,M') &= D(i-1,M') \vee D(i-1,M' \setminus col(\child(\M)[i])) \text{ if } i \leq t, M' \subseteq M.
\end{align*}
The algorithm returns $True$ iff $D(t,M) = True$. Informally speaking, the first part of the computation of $D(i,M')$ ignores the $i$-th child of $\M$ while the second part add this child into the potential solution.

The time and space complexities of the dynamic programming are $\bigo(2^k|V|)$ since $D$ is of size at most $2^k|V|$ and the computation time for each element is constant. Therefore, since the dynamic programming is launched in the worst case on each node of $\T(G)$, the whole time complexity is $\bigo(2^k|V|^2)$.

It remains to show the correctness of the dynamic programming. Suppose the existence of a module $\M'$ such that $col(\M') = M$. Then, either $\M'$ is a strong module represented in a node of $\T(G)$, or it is a union of $j$ modules $\M'_1,\M'_2,\dots,\M'_j$, children of a module $\M$. Therefore, $M \setminus \{ \{ col(\M'_1) \cup \{col(\M'_2)\} \cup \dots \cup \{col(\M'_j)\} \}\} = \{0,0,\dots,0\}$, then $D(t,\M) = True$. 

Conversely, if there is a module $\M$ such that $D(t,\M) = True$, then there is a union of the children of $\M$ such that the set of colors of these children is equal to $M$. \qed
\end{proof}

\begin{corollary} 
\pmogm is in $\fpt$ when parameterized by $(k,|C|)$.
\end{corollary}

\begin{proof}
Note that, by definition of the motif $M$, for each color $c \in C$, 
the number of occurrences of $c$ in $M$ is at most $k$.
Thus, the number of multisets $M'$ such that $M' \subseteq M$ is less than $k^{|C|}$. The time complexity of the algorithm in Proposition~\ref{prop:mmogm} is bounded by $\bigo(k^{|C|}|V|^2)$. 
\end{proof}

This corollary, implying a polynomial-time algorithm when the number of colors is a constant, is quite surprising and shows a fundamental difference with \pegm. Indeed, recall that this last is $\np$-complete, even when the motif is built over two different colors~\cite{Fellows2007}. 

Let us now show that even when a set of colors is associated to each node of the graph, the problem is still in the $\fpt$ class. It is indeed biologically relevant to consider many functions for a same reaction in a metabolic network or to consider more than one homology for a protein in a PPI network~\cite{Lacroix2006,Betzler2008}. A version of \pegm with a set of colors for each graph node as been defined, and thus, we can introduce the analogous problem \plmogm.

\PbDec{\plmogm}{A graph $G=(V,E)$, an integer $k$, a set of colors $C$, a multiset $M$ over $C$, a function $col : V \to 2^C$ giving a set of colors for each node of $V$.}{Does there exist a subset $V' \subseteq V$ such that (i) $|V'| = k$, (ii) $V'$ is a module of $G$ and (iii) there is a bijection $f : V' \to M$ such that $\forall v \in V', f(v) \in col(v)$.}

\begin{proposition}\label{prop:module brute force} \plmogm is in the $\fpt$ class.
\end{proposition}
\begin{proof}

We first build the modular tree decomposition $\T(G)$ from $G$. We repeat the following algorithm for each node $\M$ of $\T(G)$.

If $\M$ has less than $k$ nodes, we look for a bijection between the colors of $\M$ and $M$. To do so, we try all the possible combinations. In the worst case, there are $c^k$ such combinations, where $c$ is the number of different colors in $M$ (thus $c \leq k$). 

In the following, we thus can consider $\M$ with more than $k$ nodes. If it is a prime node, we can ignore it since this node cannot be a solution for a motif of size $k$. Otherwise, it is a series or parallel node, and a union of the children can be a solution. Let us now show that the number of possible solutions is exponential only with $k$, and it is thus possible to try all the possibilities.

To do so, we first give a bound to the number of children for $\M$. There are at most $k$ nodes in each child of $\M$ (otherwise, this child cannot be in a solution). In each child of $\M$, there are at most $2^c$ different sets of colors associated to each node. Since there are at most $k$ nodes in each child of $\M$, there are at most $(2^c)^k$ different children of $\M$. A same child of $\M$ cannot occurs more than $k$ times (otherwise, the next occurrences cannot be in a solution for a motif of size $k$). Therefore, there are at most $k(2^c)^k$ children to consider for $\M$. 

We bounded the number of children for $\M$. We now choose the potential union of children of $\M$ in the solution -- we must choose $i$ among the $k(2^c)^k$ children, where $i$ goes from $1$ to $k$. This is bounded by $(k(2^c)^k)^{k+1}$. Finally, for each union of chosen children, there are $c$ possible colors for the nodes (there are at most $k$ of them), which lead to at most $c^k$ tests. The overall complexity of the algorithm is thus exponential only in $k$.\qed
\end{proof}

\subsection{Open problems}

\paragraph{Around the module properties}
%
%
%
%
%
%

Del~Mondo \textit{et al.} \cite{Del2009} refine the modular tree decomposition with the \textit{homogeneous decomposition}, introducing two new types of nodes for the prime modules. Can we use this new structure for the \pmogm problem?
%
%
%
%
%
%
%

A \textit{split}~\cite{Cunningham1982} generalize a module. 
%
It is also possible to decompose a graph in its splits in polynomial time~\cite{Dahlhaus2000}. Can we use splits to generalize the \pmogm problem?

\paragraph{Algorithmic questions}

It would also be interesting to know if \pmogm is $\wone$-hard if the parameter is the number of colors in the motif as for \pegm, or if using modularity change its complexity class.
%
%
%
%





\bibliographystyle{spmpsci}      
\bibliography{b_short}

\end{document}